\documentclass[letterpaper]{dae-handout}
\graphicspath{{./}{./FIGURES/}{./figures/}{../Figures/}{../../Figures/}{./images/}{./graphics/}}

\newif\ifarxiv
\arxivtrue

\newcommand{\trversion}{v2.0}
\newcommand{\trfilename}{oae-cap-tufte.tex}
\newcommand{\trdate}{02026-OCT-11}
\newcommand{\trshorttitle}{Circumventing CAP with OAE}
\newcommand{\trauthor}{Paul Borrill}
\newcommand{\traffiliation}{D\AE D\AE LUS}

\usepackage[T1]{fontenc}
\usepackage[utf8]{inputenc}
\usepackage[bitstream-charter]{mathdesign}
\DeclareSymbolFont{operators}   {OT1}{cmr} {m}{n}
\DeclareSymbolFont{letters}     {OML}{cmm} {m}{it}
\DeclareSymbolFont{symbols}     {OMS}{cmsy}{m}{n}
\DeclareSymbolFont{largesymbols}{OMX}{cmex}{m}{n}
\SetSymbolFont{operators}{bold} {OT1}{cmr} {bx}{n}
\SetSymbolFont{letters}  {bold} {OML}{cmm} {b}{it}
\SetSymbolFont{symbols}  {bold} {OMS}{cmsy}{b}{n}
\SetMathAlphabet{\mathit} {normal}{OT1}{cmr}{m}{it}
\SetMathAlphabet{\mathbf} {normal}{OT1}{cmr}{bx}{n}
\SetMathAlphabet{\mathsf} {normal}{OT1}{cmss}{m}{n}
\SetMathAlphabet{\mathtt} {normal}{OT1}{cmtt}{m}{n}
\usepackage{amsmath}
\usepackage{amsthm}
\usepackage{booktabs}
\usepackage{enumitem}
\usepackage{fancyhdr}
\usepackage{titletoc}
\usepackage{etoc}
\usepackage{lastpage}
\usepackage{xcolor}
\usepackage{tikz}
\usepackage{eso-pic}
\usepackage{url}
\usepackage{morefloats}   
\usepackage{placeins}     
\usetikzlibrary{positioning,shapes.geometric}

\setcitestyle{numbers,square}

\newtheorem{theorem}{Theorem}[section]
\newtheorem{lemma}[theorem]{Lemma}
\newtheorem{proposition}[theorem]{Proposition}
\newtheorem{corollary}[theorem]{Corollary}
\newtheorem{definition}[theorem]{Definition}

\definecolor{linkblue}{RGB}{30,80,180}
\definecolor{linkgreen}{RGB}{30,120,60}
\definecolor{linkred}{RGB}{140,30,30}

\hypersetup{
  colorlinks=true,
  urlcolor=linkblue,
  linkcolor=linkgreen,
  citecolor=linkgreen,
}

\makeatletter
\ifarxiv
  \newcommand{\sidecite}[3][]{\citep{#2}}
\else
  \newcommand{\sidecite}[3][]{%
    \def\sc@temp{#1}%
    \citep{#2}%
    \ifx\sc@temp\empty
      \marginnote{\scriptsize\citep{#2}\enspace #3}%
    \else
      \marginnote[#1]{\scriptsize\citep{#2}\enspace #3}%
    \fi
  }
\fi
\makeatother

\tikzset{
  badge/.style={
    circle,
    draw=#1,
    fill=#1!10,
    line width=0.5pt,
    minimum size=0.45in,
    font=\tiny\sffamily,
    align=center,
    text=#1!80!black
  }
}

\newcommand{\placebadges}{%
  \ifarxiv\else
  \AddToShipoutPictureBG*{%
    \AtPageUpperLeft{%
      \raisebox{-0.5in}{\hspace{\dimexpr\paperwidth-0.7in\relax}%
        \begin{tikzpicture}[overlay]
          \node[inner sep=0pt] (logo)
            {\includegraphics[height=0.855in]{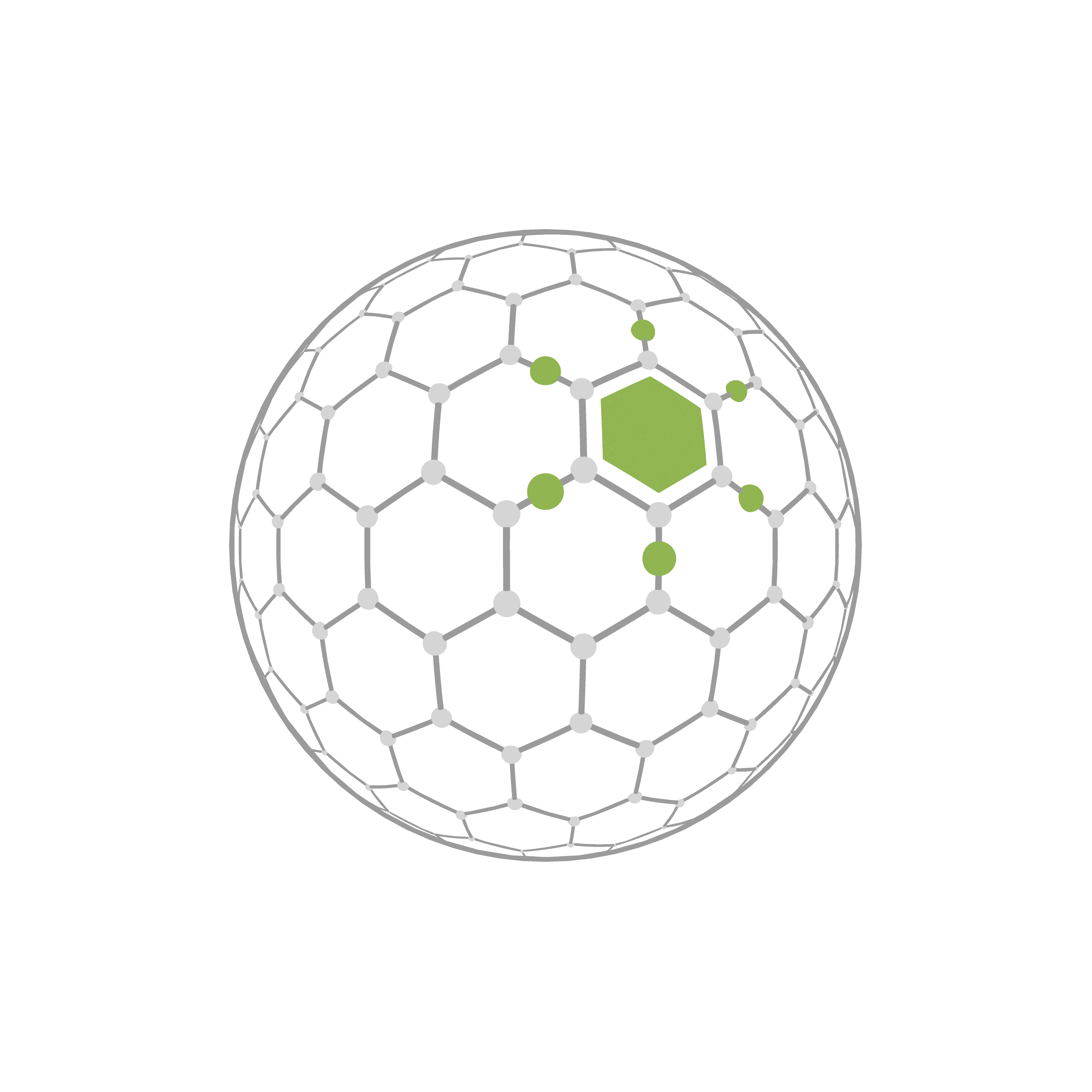}};
          \node[badge=green!60!black, left=0.08in of logo]  (b1) {Artifacts\\Available};
          \node[badge=purple!70!black, left=0.08in of b1]   (b2) {Expert\\Verified};
          \node[badge=green!50!black,  left=0.08in of b2]   (b3) {AI\\Assisted};
          \node[badge=blue!70!black,   left=0.08in of b3]   (b4) {Human\\Conceived};
        \end{tikzpicture}%
      }%
    }%
  }%
  \fi
}

\newcommand{\maketrcover}{%
  \ifarxiv\else
  \thispagestyle{empty}
  \begin{fullwidth}
  \vspace*{2in}
  \begin{center}
    {\Large\sffamily\bfseries D\AE D\AE LUS Technical Report}\\[1.5em]
    {\LARGE\sffamily\bfseries Circumventing the CAP Theorem\\[0.3em]
     with Open \AE thernet}\\[2em]
    {\large \trauthor\,,\;\traffiliation}\\[1em]
    {\normalsize \trversion\quad---\quad\trdate}
  \end{center}

  \vspace{2em}
  \noindent\rule{\linewidth}{0.4pt}

  \vspace{1em}
  \footnotesize
  \begin{description}[leftmargin=1.2in, style=sameline, font=\normalfont\scshape]
    \item[Status:]       v2.0 --- corrected arXiv replacement (v1.0: 02026-FEB-24)
    \item[Filename:]     \texttt{\trfilename}
    \item[Keywords:]     CAP theorem, partition tolerance, Open \AE thernet, OAE,
                         bisynchrony, Shannon Slots, octavalent mesh, CAL theorem
    \item[Related:]      Circumventing FLP (arXiv, 2026),
                         OAE Specification (OCP),
                         Leibniz Bridge Synthesis (DAE internal)
    \item[License:]      \textcopyright\ 2026 \trauthor, \traffiliation.
                         All rights reserved.
  \end{description}

  \vspace{1.5em}
  \noindent\rule{\linewidth}{0.4pt}

  \vspace{2em}
  \begin{center}
    \normalsize\itshape
    This cover page may be discarded when printing.\\
    The paper begins on the following page.
  \end{center}

  \end{fullwidth}
  \clearpage
  \fi
}

\fancypagestyle{plain}{%
  \fancyhf{}%
  \fancyfoot[L]{\small\textsc{\trshorttitle}}%
  \fancyfoot[C]{\small\thepage\ of \pageref{LastPage}}%
  \ifarxiv\else
  \fancyfoot[R]{\scriptsize\texttt{\trfilename}}%
  \fi
  \ifarxiv\else
  \fancyfoot[L]{\raisebox{-0.8in}{\tiny\texttt{\trversion\ -- \trdate}}}%
  \fancyfoot[R]{\raisebox{-0.8in}{\tiny\texttt{\trfilename}}}%
  \fi
}
\title{Circumventing the CAP Theorem\texorpdfstring{\\}{ }with Open \AE thernet}
\author{Paul Borrill, D\AE D\AE LUS}
\date{02026-FEB-24; revised 02026-OCT-11 (v2)}

\begin{document}
\maketrcover
\setcounter{page}{1}
\maketitle
\hypersetup{pdftitle={Circumventing the CAP Theorem with Open \AE thernet}}
\placebadges
\thispagestyle{plain}

\daemargintoc

\begin{abstract}
The CAP theorem is routinely treated as a systems law: under network partition, a replicated service must sacrifice either consistency or availability. The theorem is correct within its standard asynchronous network model, but operational practice depends on where partition-like phenomena become observable and on how lower layers discard or preserve semantic information about message fate. This paper argues that Open \AE thernet (OAE) shifts the engineering regime in which CAP tradeoffs become application-visible by (i)~replacing fire-and-forget link semantics with bounded-time bilateral reconciliation of endpoint state---the property we call \emph{bisynchrony}---and (ii)~avoiding Clos funnel points via an octavalent mesh in which each node can act as the root of a locally repaired spanning tree. The result is not the elimination of hard graph cuts, but a drastic reduction in the frequency and duration of application-visible ``soft partitions'' by detecting and healing dominant fabric faults within hundreds of nanoseconds, under the link-fault model and repair rule stated in the paper. We connect this view to Brewer's original CAP framing, the formalization by Gilbert and Lynch, the CAL theorem of Lee et~al., which replaces binary partition tolerance with a quantitative measure of apparent latency, and Abadi's PACELC extension.
\end{abstract}

\FloatBarrier
\section[Introduction]{Introduction}

Brewer's CAP theorem frames a central tension in distributed systems design: when networks partition, systems must trade off consistency and availability\sidecite[2in]{Brewer00}{Brewer, \emph{PODC} Keynote, 2000}. Gilbert and Lynch provided a widely cited formalization in an asynchronous network model, proving that for a particular notion of atomic consistency and availability, partitions force a choice\sidecite{GilbertLynch02}{Gilbert \& Lynch, \emph{SIGACT News}, 2002}. Over time, CAP has expanded from a scoped impossibility result into a cultural slogan that shapes architecture: the network is unreliable; timeouts are inevitable; retries are normal; and weakened consistency is the price of scale.


This paper makes a narrower claim. We do not argue CAP is false, nor that physical partitions can be prevented. Rather, we argue that in contemporary systems the practical bite of CAP depends critically on two design decisions that are commonly treated as immutable:

\begin{enumerate}[leftmargin=*, itemsep=3pt]
\item \textbf{Where semantic ambiguity is created.} Conventional link layers discard information about message fate and force higher layers to infer failure by timeout.
\item \textbf{When partitions become observable.} A partition is operationally declared when latency exceeds a detection threshold, not when physics first deviates.
\end{enumerate}

Open \AE thernet (OAE) changes both decisions. It moves decisive coordination down to the link and fabric substrate, making most fault detection and repair local, bounded-time, and below the observation threshold of higher-layer protocols.

\marginnote{\footnotesize This paper is a companion to ``Circumventing the FLP Impossibility Result with Open Atomic Ethernet''~\citep{Borrill18}, which addresses the Fischer--Lynch--Paterson result using the same intellectual framework.}

\FloatBarrier
\section[What CAP Actually Says]{What CAP Actually Says (and What It Does Not Say)}

CAP is often summarized as ``pick two of consistency, availability, and partition tolerance.'' Brewer's keynote is more careful: the tension arises when partitions occur, and the tradeoff is fundamentally about \emph{behavior under partition}\sidecite{Brewer00}{Brewer, \emph{PODC} Keynote, 2000}. Gilbert and Lynch formalize the claim under an asynchronous network model in which messages may be delayed arbitrarily or lost, and they define availability as the requirement that every request to a non-failed node eventually returns a response\sidecite{GilbertLynch02}{Gilbert \& Lynch, \emph{SIGACT News}, 2002}.

Two clarifications are essential:

\begin{itemize}[leftmargin=*, itemsep=3pt]
\item \textbf{CAP does not claim partitions are preventable or avoidable.} Hard graph cuts remain possible---cable cuts, power isolation, physical separation.
\item \textbf{CAP does not require that ``partition'' be a binary event.} Operationally, systems experience a spectrum: transient link faults, microbursts, queue collapse, control-plane reconvergence, and tail latency excursions. Many of these present as partitions to software because the lower layers provide ambiguous delivery states and slow failure detection.
\end{itemize}

\marginnote{\footnotesize Brewer himself revisited these nuances in ``CAP Twelve Years Later''~\citep{Brewer12}, acknowledging that the ``pick two'' framing had become misleading and that partitions are rare events requiring specific strategies.}

The most common category mistake is to treat CAP as a claim about nature---``the network is inherently partitionable''---rather than as a theorem about a modeling choice: asynchronous communication with ambiguous silence and unilateral message ownership.

\FloatBarrier
\section[The Formal Model]{The Gilbert--Lynch Formal Model}

Understanding what OAE circumvents requires understanding precisely what Gilbert and Lynch prove. Their formalization uses an asynchronous network model with the following properties:

\begin{enumerate}[leftmargin=*, itemsep=3pt]
\item \textbf{Asynchronous communication:} There is no upper bound on message delivery time. A message may take any finite amount of time to arrive, or may be lost entirely.
\item \textbf{Unilateral message passing:} Communication is by one-way send. The sender transmits and hopes; the receiver either gets the message or does not. There is no atomic bilateral exchange.
\item \textbf{Linearizable consistency:} The system provides a read-write register with linearizability---every operation appears to take effect atomically at some point between invocation and response.
\item \textbf{Availability as liveness:} Every request received by a non-failing node must eventually receive a response.
\item \textbf{Partition as message loss:} A partition is modeled as the complete loss of all messages between two disjoint sets of nodes for an indefinite period.
\end{enumerate}

\marginnote[-1cm]{\footnotesize Kleppmann~\citep{Kleppmann15} identified several issues with this formalization, including its narrow scope (single read-write register), definitional ambiguity, and binary framing that oversimplifies graduated tradeoffs.}

The proof is a simple two-node argument. During a partition, node~1 receives a write and node~2 receives a read. For availability, both must respond. For consistency (linearizability), the read must return the write's value. But the partition prevents communication. Contradiction.

The proof's power comes from its generality within the asynchronous model. Its limitation comes from the model itself: it assumes that partitions are total, indefinite, and detected only by the absence of messages.

\FloatBarrier
\section[From CAP to CAL]{From CAP to CAL: Quantifying Partition as Apparent Latency}

Lee et~al.\ generalize CAP by replacing the binary notion of partition tolerance with a numerical measure of \textbf{apparent latency}~$L$ and deriving an algebraic relation between inconsistency, unavailability, and~$L$ (the CAL theorem)\sidecite{LeeCAL21}{Lee et~al., \emph{arXiv:2109.07771}, 2021}. In the CAL view, CAP is recovered as a limiting case: if latency becomes unbounded (the network is effectively partitioned), then one of inconsistency or unavailability must also become unbounded.

\marginnote{\footnotesize The CAL theorem replaces binary properties with continuous measures: inconsistency as a time interval (maximum data staleness), unavailability as a time interval (maximum response delay), and apparent latency including network delay, execution overhead, and clock synchronization error.}

This matters for OAE because it makes explicit what engineering practice already does implicitly:

\begin{itemize}[leftmargin=*, itemsep=3pt]
\item A ``partition'' is operationally declared when latency exceeds a timeout threshold.
\item Timeouts are not physical constants; they are guesses made necessary by semantic ambiguity.
\end{itemize}

Thus, changing where ambiguity is created and reducing the tail of apparent latency at the substrate directly changes the regime in which CAP-like tradeoffs become visible. In CAL terms, the question is not whether partitions are possible, but whether apparent latency can be kept bounded for the dominant class of faults.

\FloatBarrier
\section[The Category Mistake]{The Category Mistake in Conventional Networking}

The deeper question is: why does conventional networking operate in the asynchronous model at all? The answer is historical, not physical.

Ethernet was designed as a shared-medium broadcast protocol where collisions were expected and frame loss was normal. When switched Ethernet replaced shared media, the point-to-point links between switch and host became deterministic physical channels---but the protocol semantics inherited from the broadcast era were never updated. Layer~2 remained fire-and-forget: frames are transmitted into the void, with no bilateral acknowledgment, no transactional semantics, and no guaranteed delivery.

\marginnote{\footnotesize The ``fire-and-forget'' semantics of Layer~2 Ethernet are a historical artifact of CSMA/CD, not a physical necessity of point-to-point links. See the companion FLP paper~\citep{Borrill18} for the full development of this argument.}

This creates an \emph{artificial} asynchrony. The physical link between two directly connected endpoints has bounded propagation delay (determined by the speed of light in the medium and the cable length), bounded serialization delay (determined by the line rate and frame size), and deterministic error detection (via CRC). The physics provides bounded propagation and serialization delay. It is only the protocol that imposes asynchrony by choosing not to exploit this structure.

The consequence is that every layer above Layer~2 inherits the asynchronous model by default. TCP must deal with unbounded delays, lost segments, and reordering---not because the physics demands it, but because Layer~2 discards the acknowledgment information that would prevent these pathologies. Higher layers then spend enormous complexity reconstructing, approximately and probabilistically, what the physical layer knew all along.

This is the category mistake: treating a semantic property (transactional atomicity) as something that can be constructed from syntactic mechanisms (checksums, sequence numbers, retransmission timers) layered over a substrate that has already destroyed the relevant information.

\FloatBarrier
\section[OAE's Model]{OAE: Bisynchronous Link Semantics and Slot Reconciliation}\label{sec:OAEmodel}

OAE begins from a different substrate assumption: a directly connected point-to-point link provides bounded propagation and serialization delay. Conventional Ethernet semantics do not exploit this; they retain a broadcast-era fire-and-forget ontology in which a frame may be ``in flight'' with ambiguous ownership.

OAE instead treats each link as a \textbf{slot reconciliation protocol}: a paired register at each endpoint is reconciled so that each round has a bounded-time, bilateral, and unambiguous outcome. The protocol is designed so that silence within a reconciliation interval is informative rather than ambiguous.

\begin{figure*}[ht]
\centering
\includegraphics[width=\linewidth]{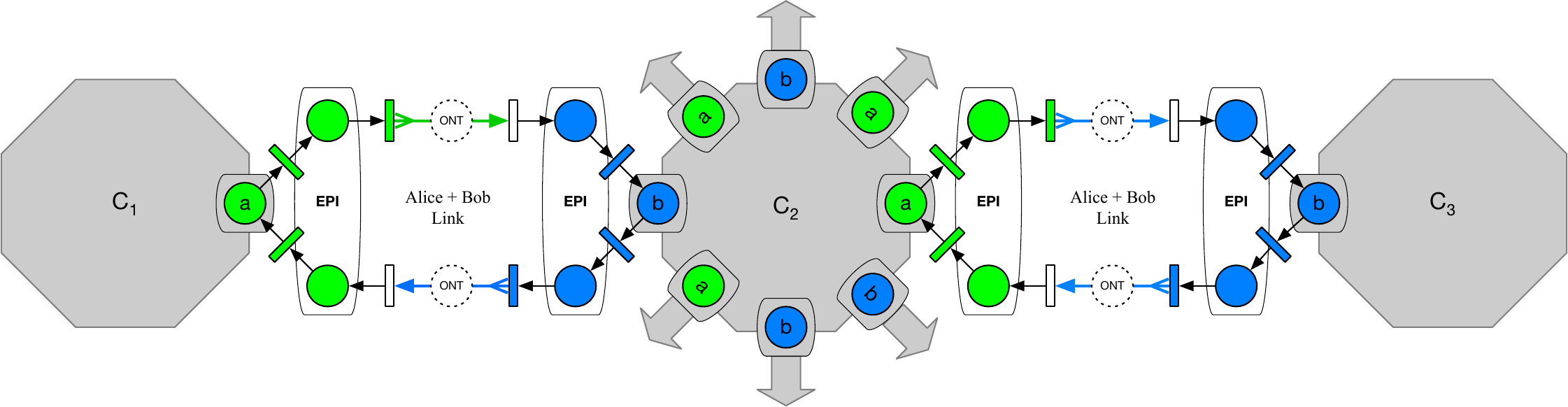}
\caption{Three OAE cells ($C_1$, $C_2$, $C_3$) connected by bilateral links. Each link endpoint contains an EPI register pair: Alice (green) and Bob (blue). Cell~$C_2$ is an octavalent cell with all eight ports visible. The link between two cells is not a one-way channel but a paired register reconciliation: both endpoints participate symmetrically in each slot, and the outcome---$(M,M)$ or $(\emptyset,\emptyset)$---is known to both parties at round boundary.}
\label{fig:alice-bob-charlie}
\end{figure*}

\subsection[Shannon Slots]{Shannon Slots: Register Reconciliation, Not Time Slots}

A critical distinction must be drawn at the outset. Slotted Aloha and Ethernet quantize \emph{time on the medium}. OAE quantizes \emph{ownership of state between endpoints}. The former solves contention; the latter solves ambiguity. These are not incremental differences---they are different problem definitions.

\marginnote[1.5cm]{\footnotesize The term ``slot'' in networking history refers to a time window. OAE's Shannon Slots are a fundamentally different object: a paired register reconciliation that atomically resolves ownership of state.}

A Shannon Slot is not a time window on a shared medium---it is a \emph{paired register} at each end of a point-to-point link. The wire is not the scarce object; \emph{ambiguity} is the scarce object. OAE is a slot reconciliation protocol: it atomically reconciles the contents of register pairs across the link.

The time bound~$\Delta$---the minimum interval for one complete bilateral reconciliation, determined by line rate and cable length---is a physical consequence of register reconciliation, not the defining abstraction. What matters is not the duration but the outcome: after each reconciliation, both endpoints know the state of both registers. This gives three properties:

\begin{enumerate}[leftmargin=*, itemsep=3pt]
\item \textbf{Bilateral resolution:} Both registers reach one of two states---$(M, M)$ or $(\emptyset, \emptyset)$---with no ambiguous intermediate ownership.
\item \textbf{Bounded completion:} Reconciliation completes within~$\Delta$, which is known and fixed.%
\marginnote{\footnotesize $\Delta$ is not a time window---it is the duration of a bisynchronous register exchange, bounded by the two-way speed of light in the medium and the serialization delay at line rate.}
\item \textbf{Informative silence:} If no register update arrives within~$\Delta$, none was sent. Silence is a definitive outcome, not an ambiguity.
\end{enumerate}

This is precisely the model under which Halpern and Moses\sidecite{HalpernMoses90}{Halpern \& Moses, \emph{JACM}, 1990} proved that common knowledge is achievable: a synchronous system with guaranteed delivery within time~$\Delta$, where common knowledge---the infinite regress of ``I know that you know that I know\ldots''---is established at round boundaries. Two conditions matter here: the round boundary itself must be a common-knowledge event (lockstep rounds), and delivery must be guaranteed. Halpern and Moses also show that bounded but \emph{uncertain} delivery time alone yields only approximate ($\varepsilon$-) common knowledge; the bisynchronous claim therefore rests on the slot boundary being shared by construction, not on the delay bound alone.

\begin{figure*}[ht]
\centering
\includegraphics[width=\linewidth]{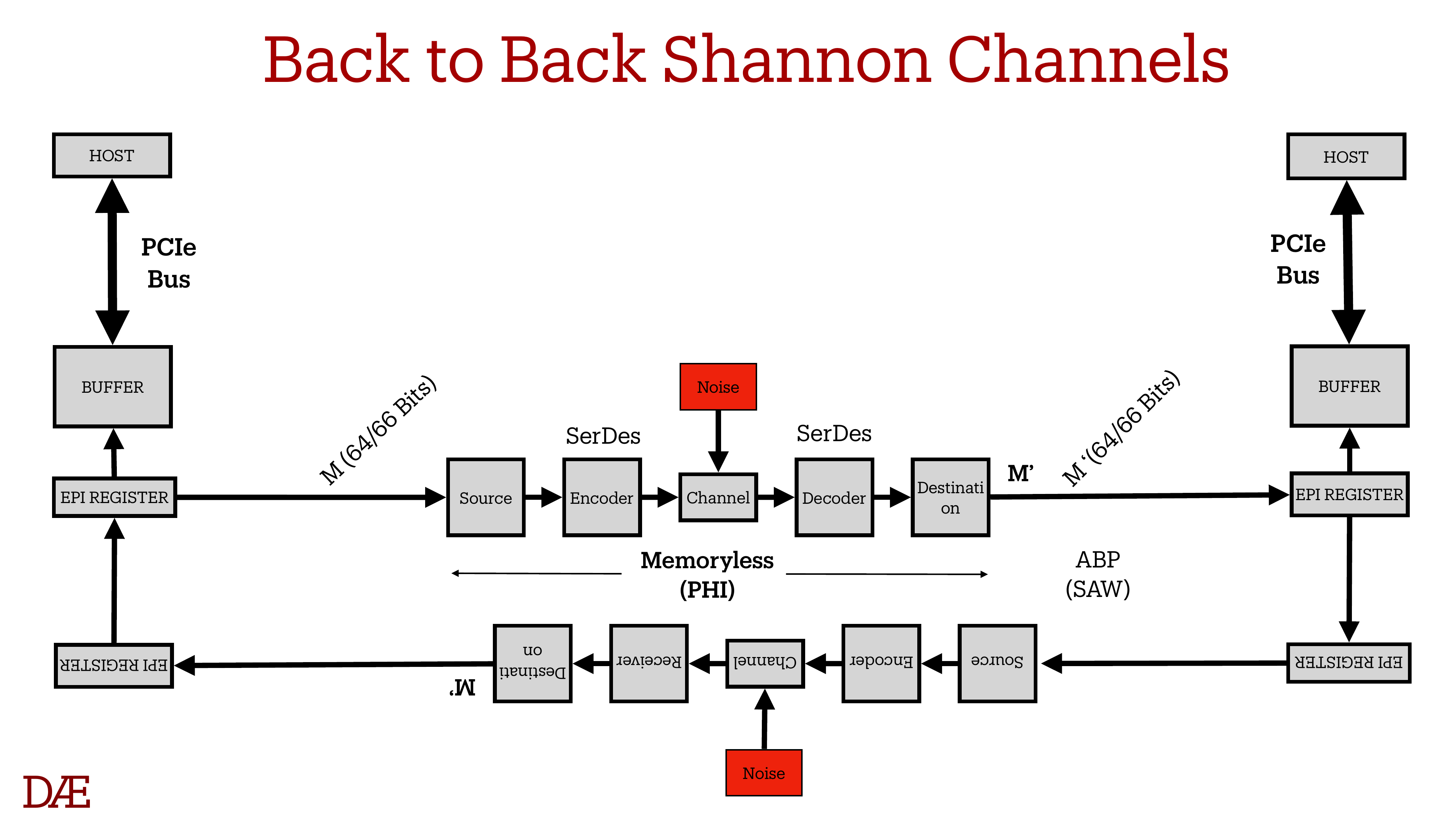}
\caption{Back-to-back Shannon channels forming a bilateral OAE link. Each direction comprises a classical Source--Encoder--Channel--Decoder--Destination chain with SerDes and 64/66-bit encoding. The memoryless (PHI) region spans the physical channel; the Alternating Bit Protocol (ABP / Stop-and-Wait) region spans the EPI register reconciliation at each endpoint. The two channels share no state---each is an independent noisy channel in the Shannon sense---but the bilateral register swap at the endpoints converts the pair into a single atomic reconciliation with common-knowledge outcome.}
\label{fig:b2b-shannon}
\end{figure*}

\subsection[Bisynchrony]{Bisynchrony: Beyond the Synchronous/Asynchronous Dichotomy}

The term \textbf{bisynchronous} denotes a strictly stronger property than bounded-time delivery: \textbf{bounded-time bilateral resolution} in which both parties reach common knowledge of outcome at each round boundary. The prefix \textbf{bi-} denotes bilateral epistemic symmetry, not dual-clock alignment.

\marginnote[3.5cm]{\footnotesize The three axes that ``synchronous/asynchronous'' incorrectly collapses: (1)~timing model, (2)~API blocking semantics, (3)~epistemic commit symmetry. OAE's innovation is on axis~(3).}

A synchronous protocol guarantees bounded delivery but permits \emph{asymmetric knowledge states}: the sender knows it sent, but not whether the receiver committed. Many application-visible partitions are not hard graph cuts; they are inference failures caused by ambiguous silence and unbounded apparent latency induced by timeout-based detection. Bisynchrony eliminates the asymmetric knowledge gap that drives these failures.

\begin{description}[leftmargin=*, itemsep=3pt]
\item[Asynchronous:] Unbounded delay, ambiguous silence, no knowledge of peer state.
\item[Synchronous:] Bounded delay, silence is informative, but knowledge may be asymmetric (sender vs.\ receiver).
\item[Bisynchronous:] Bounded delay, silence is informative, and both parties reach common knowledge of outcome at round boundary.
\end{description}

In the failure-detector framework of Chandra and Toueg\sidecite{CT96}{Chandra \& Toueg, \emph{JACM}, 1996}, conventional systems rely on unreliable failure detectors ($\Diamond\mathcal{P}$ or $\Diamond\mathcal{S}$) whose accuracy depends on timeout tuning, which in turn depends on the tail of the latency distribution---precisely the quantity that soft partitions corrupt. OAE's bisynchronous link semantics replace timeout-based failure inference with deterministic slot outcomes: a link is either reconciled or failed at each slot boundary, with no intermediate ambiguity. Under assumptions A1--A4 below, this is operationally equivalent to a \emph{perfect} failure detector for link faults---the strongest class in the Chandra--Toueg hierarchy---obtained from substrate design rather than from timeout tuning. Note the scope: it classifies the \emph{link} outcome of each slot; from one endpoint, a failed link and a crashed peer remain the same observation, so it is not a perfect detector of remote node failure.

\paragraph{Assumptions and scope.} The claims in this paper about bilateral outcomes, the elimination of category~$U$ (Section~\ref{sec:HSU}), and perfect link-fault detection are conditional on the following properties of the OAE reconciliation protocol: (A1)~\emph{bounded reconciliation}: each slot completes or is declared failed within a known~$\Delta$; (A2)~\emph{shared slot boundary}: both endpoints evaluate the outcome at the same, mutually known boundary; (A3)~\emph{agreement}: both endpoints reach the same outcome, $(M,M)$ or $(\emptyset,\emptyset)$, including under frame loss and asymmetric (one-direction) faults; (A4)~\emph{fail-stop link faults}: a faulty link stops reconciling rather than reconciling incorrectly. This paper specifies the permitted outcome set; it does not give the reconciliation state machine or prove agreement and termination under crash and asymmetric faults. Those proofs belong to the OAE specification and companion work, and the claims here hold only in the domain where A1--A4 are established.

\subsection[Swap Not Send]{Swap, Not Send: The Universal Primitive}

OAE does not \emph{send} messages. It performs \emph{bilateral register swaps}: atomic exchanges between NIC register pairs at each end of the link. In a conventional ``send'' operation, the sender unilaterally transmits a frame. The frame enters an ambiguous intermediate state---``in the wire''---where ownership is undefined. In an OAE register swap, both endpoints participate symmetrically. The reconciliation has only two permitted outcomes:

\begin{itemize}[leftmargin=*, itemsep=3pt]
\item $(M, M)$: Both registers hold the message. Transaction committed.
\item $(\emptyset, \emptyset)$: Neither register holds the message. Transaction aborted.
\end{itemize}

No intermediate state is permitted. This connects directly to Herlihy's consensus hierarchy\sidecite{Herlihy91}{Herlihy, \emph{TOPLAS}, 1991}: memory-to-memory swap has infinite consensus number---it is a universal primitive. OAE builds its coordination on a universal primitive rather than the consensus-number-1 read/write operations that conventional networking employs.

\subsection[Soft vs Hard Partitions]{Soft Partitions vs.\ Hard Partitions}

We use the following operational distinction:

\begin{itemize}[leftmargin=*, itemsep=3pt]
\item \textbf{Hard partition:} a graph cut that prevents communication between components for an extended duration---physical separation, complete path loss.
\item \textbf{Soft partition:} a transient loss of reachability or decisive knowledge caused by local faults, congestion-induced drops, or control-plane reconvergence, which is observed by software as a partition because detection is timeout-based and repair is slow.
\end{itemize}

\marginnote{\footnotesize The ECC memory analogy: bit flips occur but are corrected before any software observes them. The memory appears perfect because correction is below the observation threshold. OAE applies the same principle to link and fabric faults.}

OAE does not claim to eliminate hard partitions. Its claim is that it dramatically reduces the frequency and duration of soft partitions by shifting detection and repair below the observation threshold of higher layers. For a detailed case study of how soft partitions manifest as user-visible failures in a production cloud system, see~\citep{BorrilliCloud26}.

\FloatBarrier
\section[Fabric Consequence]{Fabric Consequence: Octavalent Mesh Instead of Clos Funnel Points}

In contemporary data centers, Clos (fat-tree) fabrics concentrate forwarding and failure semantics into tiered switches. A link or switch fault can trigger control-plane reconvergence and rerouting, often taking microseconds to seconds, and can propagate disruption globally.

\begin{figure}[h!]
\centering
\includegraphics[width=\linewidth]{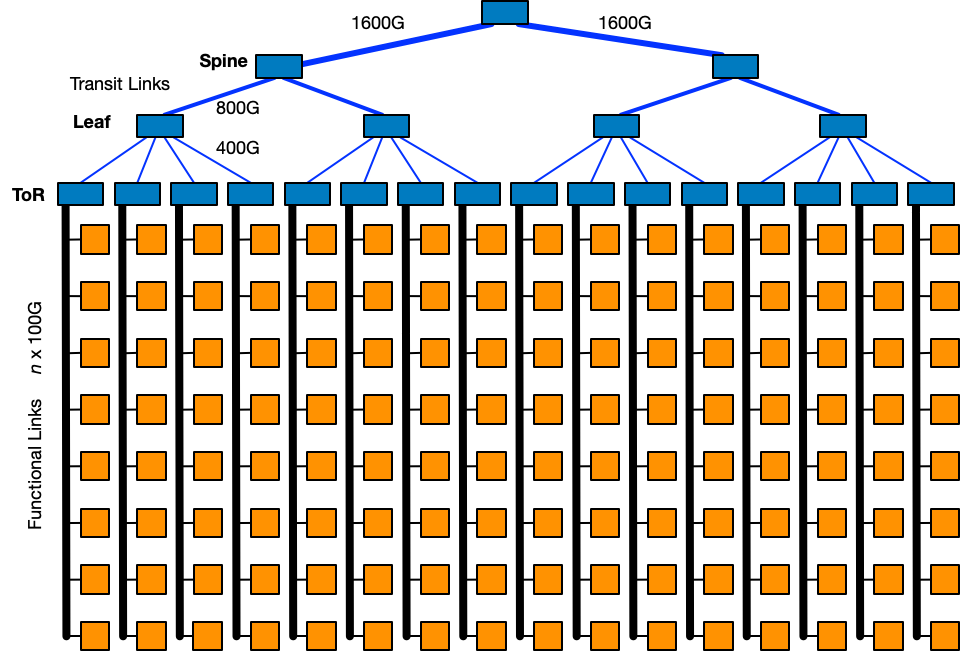}
\caption{Conventional three-tier Clos fabric (logical view). Traffic between compute nodes (orange) traverses ToR, Leaf, and Spine tiers via 100G transit links. A single switch or link failure triggers global control-plane reconvergence---the window during which higher layers observe a partition.}
\label{fig:clos}
\end{figure}
\begin{figure}[h!]
\centering
\includegraphics[width=\linewidth]{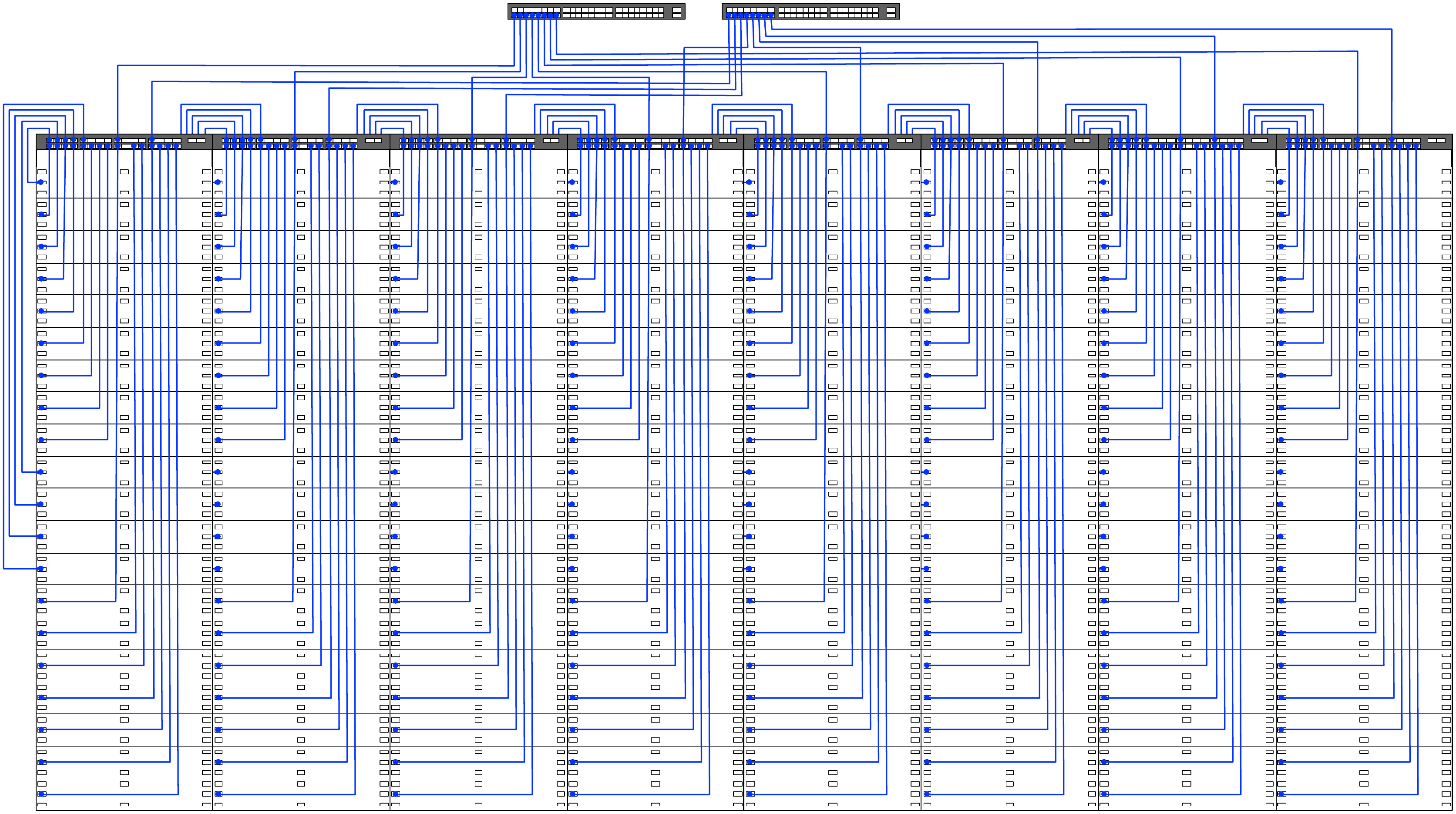}
\caption{The same Clos fabric (physical wiring). Individual server NICs are cabled point-to-point to Top-of-Rack switches via SFP+/QSFP transceivers over copper or fiber. Uplinks (blue cables) fan out to leaf and spine tiers above. Every link---NIC to ToR, ToR to leaf, leaf to spine---is a unilateral fire-and-forget channel: the sender transmits a frame and receives no link-layer acknowledgment of delivery.}
\label{fig:clos-physical}
\end{figure}

The physical reality behind the Clos diagram is shown in Figure~\ref{fig:clos-physical}. Each server NIC is individually cabled to a Top-of-Rack (ToR) switch, typically via a dedicated point-to-point link (SFP+ or QSFP transceiver over copper or fiber). Every one of these links is a fire-and-forget channel: the NIC transmits a frame and receives no link-layer acknowledgment of delivery. The ToR switch then fans traffic upward through leaf and spine tiers over additional point-to-point links. Each hop introduces another instance of the same unilateral semantics---frames dispatched into a forwarding pipeline with no bilateral resolution of outcome. A single cable pull, transceiver failure, or switch reboot anywhere in this tree triggers control-plane reconvergence that propagates through the entire fabric.

\begin{figure}[ht]
\centering
\includegraphics[width=\linewidth]{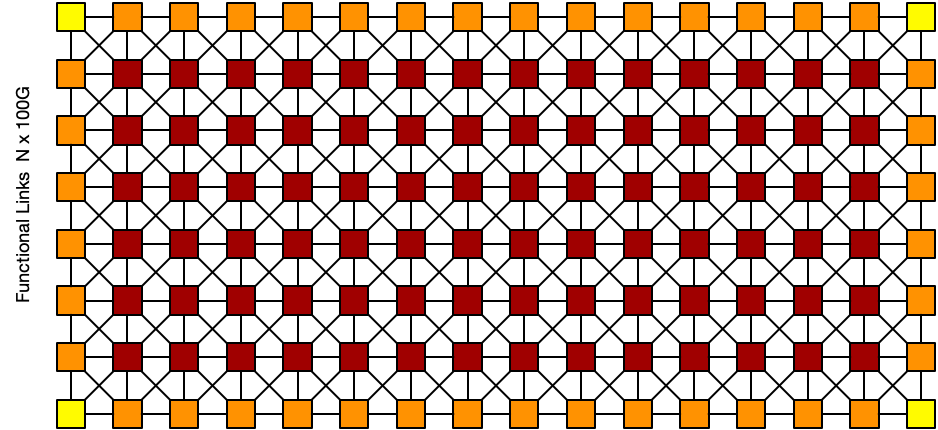}
\caption{OAE Cellular Fabrix: a $10\times 20$ octavalent mesh with valency color-coding. Yellow cells (corners) have valency~3; orange cells (edges) have valency~5; dark red cells (core) have valency~8. Each functional link operates at 100G. Every cell is both compute and forwarding---there are no proprietary switch tiers. A single link failure is healed locally by parent reselection, not by global control-plane reconvergence. Note that this is a scale-independent architecture: the sea of nodes/XPUs can extend in the planar direction (north, east, west, south) via connections to the orange cells, eliminating the need for hierarchical switches.}
\label{fig:cellular-fabrix}
\end{figure}

OAE instead assumes an \textbf{octavalent mesh} in which:

\begin{itemize}[leftmargin=*, itemsep=3pt]
\item Each cell connects to eight neighbors.
\item Every cell forwards packets---no proprietary funnel points.
\item Every node can be the root of its own spanning tree, and local repair chooses alternate parent links toward the root using only local information.
\end{itemize}

\marginnote{\footnotesize In an $n \times n$ octavalent mesh, the number of spanning trees grows exponentially in the number of nodes~$|V|$ (by Kirchhoff's matrix tree theorem). Even a modest mesh has millions of alternate trees.}

When a link fails, the sequence under Clos is: (1)~loss detection via keepalive expiration or control-plane signaling, (2)~propagation of failure information through the routing control plane, (3)~recalculation of forwarding tables, (4)~installation of new forwarding state across affected switches. Even in optimized environments, this introduces a nontrivial window during which packets may be dropped, flows rerouted inconsistently, queues built asymmetrically, and endpoints forced to infer failure through timeout. From the perspective of higher layers, this window is indistinguishable from a partition.

In an OAE octavalent mesh, detection is not inferred from silence at an arbitrary timeout scale. Instead, it is resolved within a bounded number of reconciliation rounds at the link layer. Child cells switch to pre-allocated alternate parent links toward the same root using only local observer-view information. No global control-plane recomputation is required.

The essential distinction is not merely topological but semantic:

\begin{itemize}[leftmargin=*, itemsep=3pt]
\item \textbf{Clos:} failure repair is a control-plane event requiring global state propagation and recomputation.
\item \textbf{OAE mesh:} failure repair is a data-plane event resolved locally via bounded-time bilateral link semantics.
\end{itemize}

Because repair is local and bounded, the duration of disrupted reachability for dominant single-link or single-cell faults is constrained by the link reconciliation bound---hundreds of nanoseconds to a few microseconds---rather than by network-wide reconvergence time. More precisely, let $\Delta$~denote the reconciliation bound per link and $D$~the graph diameter. If healing requires full tree reconstruction, the worst case is $T_{\text{heal}} \le D \cdot \Delta$. For a planar mesh, $D = O(\!\sqrt{|V|})$, so healing time scales sublinearly with system size. OAE's local parent reselection can do better: when a child cell has a surviving neighbor that is certified to be strictly closer to the root (Section~\ref{sec:loopfree}), it re-parents using only local information in $T_{\text{heal}} = O(\Delta)$---a single reconciliation round---regardless of fabric size. When no such neighbor survives, repair requires a local re-ranking whose round count must be bounded separately; the single-round figure is not a general guarantee.

\FloatBarrier
\section[Spanning Trees and Local Healing]{Spanning Trees, Per-Node Roots, and Why Meshes Heal Locally}

We model the fabric as an undirected graph $G=(V,E)$ whose vertices are cells (XPUs / SmartNICs) and whose edges are physical links. OAE's fabric claim is that the dominant failure mode in practice is not a full graph cut, but a local link or local cell fault that can be detected and healed below the observation threshold of any Layer-4-and-above protocol. Operationally, OAE builds a spanning tree per node, rooted at that node, and performs healing by local parent re-selection using only local information, avoiding global reconvergence and switch-centric control planes.

\marginnote{\footnotesize This is the ``spanning trees as resilience'' framing: each node can be the root of its own tree, and broken links can be healed around locally to form another spanning tree.}

The purpose of this section is to make the resilience claim numeric: in a mesh, the number of distinct spanning trees grows extremely rapidly with~$|V|$, giving a combinatorial reservoir of alternate routes for local healing. In contrast, a tree topology has exactly one spanning tree (itself), hence no alternate choices without changing the physical graph.

\subsection[Definitions]{Definitions}

A \emph{spanning tree} of a connected undirected graph~$G$ is a subgraph $T\subseteq G$ that is a tree and contains all vertices. Let $\tau(G)$ denote the number of spanning trees of~$G$.

A \emph{rooted spanning tree} is a spanning tree together with a distinguished root vertex $r\in V$. For undirected graphs, the number of rooted spanning trees is $|V|\cdot \tau(G)$ (choose a spanning tree, then choose its root).

\marginnote[1cm]{\footnotesize In an OAE mesh, each node~$r$ maintains a rooted tree oriented toward~$r$ (a directed arborescence), and healing corresponds to local rewiring that restores an arborescence after an edge or vertex fault.}

The key combinatorial point is that if $\tau(G)$ is enormous, then there exist enormous numbers of alternate choices for parent edges while preserving reachability to the root.

\subsection[Trees Have No Alternate Trees]{Baseline: Trees Have No Alternate Trees}

\begin{lemma}[A tree has exactly one spanning tree]
If $G$ is a tree, then $\tau(G)=1$.
\end{lemma}

\begin{proof}
A tree is already connected and acyclic. Any spanning subgraph that remains connected and acyclic must contain exactly $|V|-1$ edges. Removing any edge disconnects the tree, and adding any edge creates a cycle. Hence the only spanning tree subgraph is~$G$ itself.
\end{proof}

This is the combinatorial core of why hierarchical fabrics (in the limiting case of a single logical tree) have globally-coupled healing: there is no alternate spanning tree to switch to; the only way to restore connectivity is global recomputation.

\subsection[Kirchhoff's Matrix--Tree Theorem]{Counting Trees in Meshes via Kirchhoff's Matrix--Tree Theorem}

\begin{theorem}[Kirchhoff / Matrix--Tree]
Let $G$ be a connected graph with Laplacian $L=D-A$, where $A$ is the adjacency matrix and $D$ is the diagonal degree matrix. Delete any one row and the corresponding column of~$L$ to obtain a cofactor matrix~$L^{(i)}$. Then
\[
\tau(G)=\det\!\bigl(L^{(i)}\bigr),
\]
and the value is independent of the choice of~$i$.
\end{theorem}

\marginnote[0.5cm]{\footnotesize Kirchhoff's theorem (1847) makes $\tau(G)$ computable and, more importantly, provides closed forms for regular product graphs such as grids.}

\subsection[Grid Formula]{Exact Formula for an $m\times n$ Grid}

Let $G_{m,n} = P_m \,\square\, P_n$ be the $m\times n$ grid graph (Cartesian product of path graphs). The Laplacian eigenvalues of~$G_{m,n}$ are
\[
\lambda_{j,k} \;=\; 4 - 2\cos\!\left(\frac{\pi j}{m}\right) - 2\cos\!\left(\frac{\pi k}{n}\right),
\]
for $0\le j\le m-1$ and $0\le k\le n-1$.
The $(j,k)=(0,0)$ eigenvalue is~$0$ (as for any Laplacian); all others are positive.

\begin{proposition}[Spanning trees of a grid]
For $m,n\ge 2$,
\[
\tau(G_{m,n})
\;=\;
\frac{1}{mn}\;
\prod_{\substack{0\le j\le m{-}1\\[2pt] 0\le k\le n{-}1\\[2pt] (j,k)\neq (0,0)}}
\biggl(\,
4 - 2\cos\!\left(\frac{\pi j}{m}\right) - 2\cos\!\left(\frac{\pi k}{n}\right)
\!\biggr).
\]
\end{proposition}

\begin{proof}
By Matrix--Tree, $\tau(G)=\det(L^{(i)})$. Equivalently, using Laplacian eigenvalues,
\[
\tau(G)\;=\;\frac{1}{|V|}\;\prod_{\ell=2}^{|V|}\lambda_\ell\,,
\]
where $\lambda_1=0$ and $\lambda_2,\dots,\lambda_{|V|}$ are the nonzero Laplacian eigenvalues. For the Cartesian product $P_m\square P_n$, the Laplacian spectrum is the pairwise sum of the spectra of the factors, yielding the stated~$\lambda_{j,k}$. Substituting $|V|=mn$ gives the formula.
\end{proof}

\subsection[Exponential Growth]{Exponential Growth: Why ``Every Node Is the Root'' Scales Resilience}

The product form above immediately implies the scaling law:

\begin{corollary}[Exponentially many alternate trees]\label{cor:exp-trees}
Let $N=mn=|V(G_{m,n})|$. Then $\log \tau(G_{m,n}) = \Theta(N)$, i.e.\ $\tau(G_{m,n})$ grows exponentially in the number of nodes. Consequently the number of rooted spanning trees grows as $N\cdot \tau(G_{m,n})$, also exponential in~$N$.
\end{corollary}

\begin{proof}
Using
\[
\log \tau(G_{m,n}) = \sum_{\substack{(j,k)\neq(0,0)}} \log \lambda_{j,k} - \log(mn),
\]
and letting $m,n\to\infty$ with bounded aspect ratio (so that both dimensions grow), the double sum is a two-dimensional Riemann sum over the unit square, giving
\[
\log \tau(G_{m,n}) = c\, mn + o(mn)
\]
for a constant $c>0$. Hence $\tau(G_{m,n}) = \exp(\Theta(N))$ with $N=mn$. For fixed width (e.g.\ $m=2$, $n\to\infty$) the sum is one-dimensional and the per-vertex constant differs, but growth remains exponential in~$N$ (for the $2\times n$ ladder, $\tau$ grows as $(2+\sqrt3)^n$).
\end{proof}

\marginnote[1cm]{\footnotesize A concrete instance: a degree-limited grid-like graph with modest size already yields $\approx 4.64\times 10^{45}$ alternative spanning trees.}

This is the combinatorial backbone of the operational claim in the OAE mesh: as the mesh grows, the space of alternate spanning trees becomes astronomically large, and local healing can select a different parent edge toward the root without consulting global routing state. The contrast with chains is stark: in a chain, reliability \emph{decreases} exponentially with length (each additional link multiplies the failure probability), whereas in a sufficiently connected mesh, the number of alternative routing structures \emph{increases} super-polynomially, allowing failure absorption without global partition.

\subsection[Edge-Disjoint Trees]{From ``Many Trees'' to Fault Survivability: Edge-Disjoint Spanning Trees}

Counting trees provides intuition, but survivability under faults is more directly captured by \emph{edge-disjoint spanning trees}---and the number of such trees is governed by the graph's \emph{edge-connectivity}~$\lambda(G)$. By Menger's theorem, $\lambda(G)$ equals the maximum number of edge-disjoint paths between any pair of vertices; by the Nash--Williams / Tutte theorem, the maximum number of edge-disjoint spanning trees~$\kappa_T(G)$ satisfies $\lfloor\lambda(G)/2\rfloor \le \kappa_T(G) \le \lambda(G)$. The lower bound is a corollary, not an identity: $K_4$ has $\lambda=3$ and $\kappa_T=2$, whereas $\lfloor 3/2\rfloor=1$. Tree \emph{counts} (Kirchhoff) and tree \emph{packings} (Nash--Williams / Tutte) are different quantities, and fault tolerance is governed by the latter and, more directly, by~$\lambda$.

\begin{definition}
Spanning trees $T_1,\dots,T_k$ are \emph{edge-disjoint} if $E(T_i)\cap E(T_j)=\emptyset$ for $i\neq j$. Let $\kappa_T(G)$ be the maximum number of edge-disjoint spanning trees in~$G$.
\end{definition}

\begin{proposition}[Failure avoidance via disjoint trees]
If $\kappa_T(G)\ge k$, then for any set $F\subseteq E$ of failed edges with $|F|<k$, there exists at least one spanning tree~$T$ of~$G$ such that $E(T)\cap F=\emptyset$.
\end{proposition}

\begin{proof}
Let $T_1,\dots,T_k$ be edge-disjoint spanning trees. Each failed edge in~$F$ can intersect at most one of the~$T_i$ (by disjointness). Thus $|F|<k$ implies at least one~$T_i$ is untouched by~$F$. That tree remains a valid spanning tree in $G\setminus F$.
\end{proof}

By the Nash--Williams / Tutte theorem, $\kappa_T(G) = \min_{\mathcal{P}} \lfloor |E(\mathcal{P})| / (|\mathcal{P}|-1) \rfloor$, the minimum taken over partitions~$\mathcal{P}$ of~$V$ into at least two parts, with $E(\mathcal{P})$ the edges joining different parts. Taking the trivial bound $|E(\mathcal{P})|\ge \lambda|\mathcal{P}|/2$ gives $\kappa_T(G) \ge \lfloor \lambda(G)/2 \rfloor$; counting edges gives $\kappa_T(G)\le \lfloor |E|/(|V|-1)\rfloor$. Both bounds require $\lambda(G)$ to be \emph{established} for the fabric in question; minimum degree alone does not supply it (see below). Where $\lambda(G)\ge 2$, fewer than $\lambda(G)$ simultaneous link failures cannot disconnect the fabric, so an alternate spanning arborescence always \emph{exists}. Existence is not repair: whether local parent reselection \emph{finds} one without loops is a protocol property, treated in Section~\ref{sec:loopfree}.

\marginnote[-2.5cm]{\footnotesize The Nash--Williams / Tutte theorem (1961/1961) is the classical characterization of edge-disjoint spanning trees. It grounds the resilience argument in formal graph theory rather than intuition.}

Interpreted operationally: if the physical fabric has edge-connectivity~$\lambda$, then up to $\lambda-1$ independent link failures can occur while a full spanning tree remains available. OAE's design goal is that, because repair is performed at child cells using only local information and each node is a legitimate root, the system re-selects among available trees at nanosecond timescales rather than waiting for millisecond reconvergence; Section~\ref{sec:loopfree} states the rule under which that local re-selection is loop-free.

\subsection[Resilience Amplification]{Resilience Amplification: From Counting to Probability}

The preceding combinatorial results become operationally decisive when we ask: \emph{what is the probability that the fabric disconnects under random edge failures?}

\begin{definition}[Edge-connectivity]
The \emph{edge-connectivity} $\lambda(G)$ of a connected graph~$G$ is the minimum number of edges whose removal disconnects~$G$. By Menger's theorem, $\lambda(G)$ equals the maximum number of edge-disjoint paths between any pair of vertices.
\end{definition}

\marginnote{\footnotesize For a tree, $\lambda=1$: every edge is a bridge. For a \emph{toroidal} (wrap-around) octavalent mesh every cell has valency~8 and $\lambda=8$. For the finite open mesh of Figure~\ref{fig:cellular-fabrix}, the valency-3 corner cells limit $\lambda$ to~3.}

For a tree, $\lambda(G)=1$---every edge is a bridge whose failure disconnects the graph. Whitney's theorem gives the ordering $\kappa(G)\le\lambda(G)\le\delta(G)$ (vertex-connectivity, edge-connectivity, minimum degree): minimum degree is an \emph{upper} bound on~$\lambda$, since cutting every link of a minimum-degree cell isolates it. A lower bound on~$\lambda$ must be proved, or computed, for each fabric. Once it is, the probabilistic consequence follows:

\begin{theorem}[Resilience amplification in meshes]\label{thm:resilience}
Let $G$ be connected with edge-connectivity $\lambda=\lambda(G)$. If each edge fails independently with probability at most~$p$, then
\[
P(G\;\text{disconnected}) \;\le\; \binom{|E|}{\lambda}\, p^{\lambda} \;\le\; |E|^{\lambda}\, p^{\lambda}.
\]
Disconnection probability scales as~$O(p^{\lambda})$. Since $\kappa_T(G)\le\lambda(G)$, the same bound holds with $k=\kappa_T(G)$ in place of~$\lambda$, but it is weaker.
\end{theorem}

\begin{proof}
$G$ disconnects only if the failed set contains an edge cut, and every edge cut has at least~$\lambda$ edges. So disconnection implies that some specific set of~$\lambda$ edges has failed. The probability that a given set of~$\lambda$ edges all fail is at most~$p^{\lambda}$, and there are $\binom{|E|}{\lambda}$ such sets; the union bound completes the proof. (Edge-disjoint spanning trees give a sufficient, not a necessary, route to the same conclusion with exponent~$\kappa_T$: each failed edge can break at most one of $\kappa_T$ disjoint trees.)
\end{proof}

The bound is a union bound: it is an upper bound, not an estimate, and it assumes independent failures. The per-link probability~$p$ in the examples below is a hypothetical per-observation-interval value, not a measured OAE rate.

\paragraph{Concrete example: chain vs.\ grid at $p=10^{-6}$.}

\marginnote[-2cm]{\footnotesize The gap of more than six orders of magnitude between chain and grid disconnection probabilities at $p=10^{-6}$ (it widens as~$p$ shrinks, since the ratio scales as~$1/p$) illustrates why the topology of the substrate---not just the protocol---determines whether soft partitions reach the application layer.}

Consider a 10-node chain (a path graph $P_{10}$) and a $3\times 3$ grid ($G_{3,3}$, 9~nodes, 12~edges). Both are small fabrics; the reliability difference is already dramatic.

The chain has $\lambda(P_{10})=1$ and $\kappa_T=1$. A single edge failure disconnects it:
\[
P_{\text{disc}}^{\text{chain}} = 1-(1-p)^9 \approx 9p = 9 \times 10^{-6}.
\]
The $3\times 3$ grid has $\lambda(G_{3,3})=2$ but only $\kappa_T(G_{3,3})=1$: two edge-disjoint spanning trees on 9~vertices would need $2\times 8=16$ edges, and the grid has~12. Disconnection nevertheless requires at least $\lambda=2$ simultaneous failures, so by Theorem~\ref{thm:resilience}
\[
P_{\text{disc}}^{\text{grid}} \le \binom{12}{2}\, p^2 = 66 \times 10^{-12} \approx 7 \times 10^{-11}.
\]
The exact value follows from enumerating all $2^{12}=4096$ failure sets: the only two-edge cuts are the four pairs that isolate a corner, and
\[
P_{\text{disc}}^{\text{grid}} = 4p^2 + 16p^3 + O(p^4) \approx 4.000016\times 10^{-12}
\]
at $p=10^{-6}$, against $P_{\text{disc}}^{\text{chain}} \approx 8.999964\times 10^{-6}$. The grid is about $2.25\times 10^{6}$ times (6.35 orders of magnitude) more reliable than the chain; even the loose union bound gives a factor above $1.3\times 10^{5}$. And this is a \emph{tiny} grid.

\paragraph{The $10\times 20$ mesh of Figure~\ref{fig:cellular-fabrix}.} With eight-neighbor wiring and open boundaries this fabric has $|V|=200$ and $|E|=712$, with cell valencies 3 (corners), 5 (edges) and 8 (core). By exact computation, $\lambda=3$, and an explicit packing of three edge-disjoint spanning trees exists, so $\kappa_T=3$ (the edge-count bound $\lfloor 712/199\rfloor=3$ is attained). The corners are the weak points: the four corner-isolating events involve disjoint edge sets, so
\[
4p^3\,(1-O(p^3)) \;\le\; P_{\text{disc}}^{10\times20} \;\le\; \binom{712}{3}\,p^3 \approx 5.99\times 10^{-11}
\]
at $p=10^{-6}$, the lower bound being about $4\times10^{-18}$. The open finite mesh therefore scales as $\Theta(p^3)$, not $O(p^4)$. If the boundary cells are wired onward---to neighboring tiles, as the caption of Figure~\ref{fig:cellular-fabrix} envisages, or by wrap-around---every cell has valency~8. For the $10\times 20$ toroidal octavalent mesh ($|E|=800$) we compute $\lambda=8$ and $\kappa_T=4$, and Theorem~\ref{thm:resilience} gives $P_{\text{disc}}\le\binom{800}{8}p^8\approx 4.0\times10^{-30}$ at $p=10^{-6}$. That is the regime in which disconnection probabilities fall below~$10^{-24}$; it is reached through~$\lambda$, and it depends on how the fabric's boundary is wired. At realistic link failure rates, this is the difference between a fabric that partitions and one that does not.

\subsection[Chain Reliability]{Chain Reliability, Flapping Links, and Local Redundancy}

We now make precise the intuitive claim illustrated in Figures~\ref{fig:chain-linear} and~\ref{fig:chain-mesh}: a single pathological link can dominate end-to-end reliability in a chain, even if the average failure probability across the fabric changes only marginally.

\begin{figure}[ht]
\centering
\includegraphics[width=\linewidth]{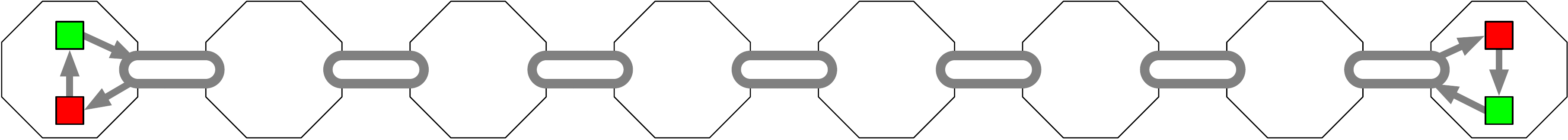}
\caption{Linear chain of OAE cells connected by bilateral AELinks. Each link succeeds independently with probability~$s_i = 1{-}p_i$. End-to-end connectivity requires \emph{every} link to survive: $P_{\text{chain}}=\prod s_i$, which decays exponentially in the number of hops. A single flapping link partitions the entire chain. Green and red indicators show committed and aborted reconciliation outcomes.}
\label{fig:chain-linear}
\end{figure}

\vspace{1.2cm}

\begin{figure}[ht]
\centering
\includegraphics[width=\linewidth]{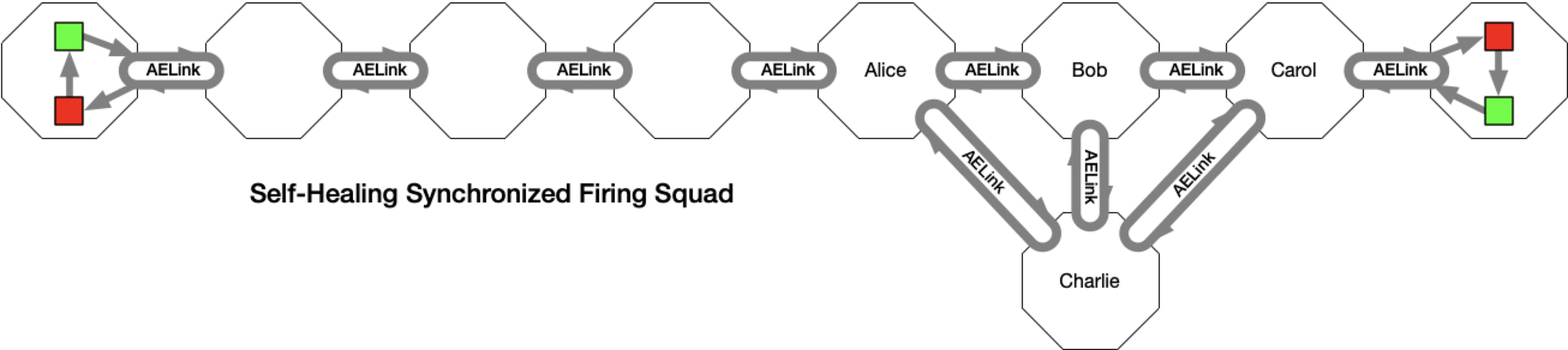}
\caption{The same chain, but with a triangle formed by Alice, Bob, and Carol through Charlie. The triangle provides a two-hop alternate path: if the direct Alice--Bob link degrades, traffic routes through Charlie. Local redundancy gives $P_\triangle = s_d + s_a - s_d\,s_a$, restoring high composite reliability even when the direct edge is degraded. Failure does not propagate globally if redundancy exists at constant graph distance.}
\label{fig:chain-mesh}
\end{figure}

\paragraph{Independent link model.}

Let a path consist of $L$ links. For link~$i$, let $p_i$ denote the probability that a reconciliation slot fails, and $s_i = 1 - p_i$ its success probability. Assuming independence per slot, the probability of successful end-to-end delivery in one attempt is
\[
P_{\text{chain}} = \prod_{i=1}^{L} s_i.
\]
If all links are healthy with identical failure probability~$p$,
\[
P_{\text{healthy}} = (1-p)^L \approx e^{-pL}
\]
for small~$p$ (the approximation needs $Lp^2\ll 1$ when $L$ is large). Even with extremely small per-link failure rates, long chains accumulate multiplicatively.

\paragraph{Single bad link amplification.}

Suppose one link becomes pathological with failure probability $q \gg p$, while the remaining $L{-}1$ links remain healthy. Then
\[
P_{\text{bad}} = (1-p)^{L-1}(1-q).
\]
The ratio relative to the healthy case is
\[
\frac{P_{\text{bad}}}{P_{\text{healthy}}}
=
\frac{1-q}{1-p}
\approx 1-q,
\]
so path reliability is dominated by its worst link. The chain is only as strong as its weakest element, independent of how many healthy links surround it.

\marginnote{\footnotesize End-to-end reliability is governed by multiplicative composition, not by average failure rate. A single bad link at $q=0.1$ collapses the chain to ${\approx}\,90\%$ reliability regardless of how many healthy links surround it.}

Yet the average failure probability across the chain is
\[
\bar{p}
=
\frac{1}{L}\sum_{i=1}^{L} p_i
=
p + \frac{q-p}{L}.
\]
For large~$L$, $\bar{p}$ changes only slightly even though the end-to-end success probability may collapse dramatically. This illustrates a general principle: end-to-end reliability is governed by multiplicative composition, not by average failure rate.

\paragraph{Flapping link model.}

\marginnote{\centering\includegraphics[height=1.2cm]{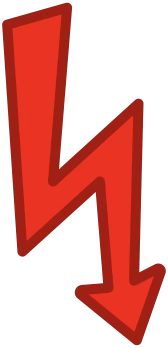}\\[2pt]\scriptsize A flapping link: the red broken arrow denotes a link oscillating between good and bad states---non-consensual link training, unstable SerDes lock, or intermittent physical faults.}

To model link flapping or non-consensual link training, consider a two-state Markov process with states $G$ (good) and $B$ (bad). Let $p_G$ and $p_B$ be the per-slot failure probabilities in each state, and let transition probabilities per slot be
\[
G \to B \;\text{with probability}\; a,
\qquad
B \to G \;\text{with probability}\; b.
\]
The stationary probability of being in the bad state is
\[
\pi_B = \frac{a}{a+b}.
\]
The effective per-slot failure probability of the link is then
\[
p_{\text{eff}} = \pi_B\, p_B + (1-\pi_B)\,p_G.
\]
Even if $p_G$ is negligible, a sufficiently large $p_B$ or transition rate~$a$ can increase $p_{\text{eff}}$ by orders of magnitude. Such behavior is frequently observed in practice as link flapping or unstable training cycles.

\paragraph{Local redundancy via triangles.}

Now consider a triangle providing two alternate routes between nodes. Let $s_d$ denote the success probability of the direct edge, and let $s_a = s_1 s_2$ denote the success probability of the two-hop alternate path. Assuming independence,
\[
P_{\triangle}
=
1 - (1-s_d)(1-s_a)
=
s_d + s_a - s_d\, s_a.
\]
Thus even if the direct edge is degraded, a healthy alternate path restores high composite reliability. This is the mathematical meaning of local healing: failure does not propagate globally if redundancy exists at constant graph distance.

\paragraph{Timeout amplification in conventional stacks.}

In conventional asynchronous stacks, failure detection is timeout-based. Let $P_{\text{chain}}$ be per-attempt success probability. If attempts are independent and identically distributed, then after $R$~retries
\[
P_{\text{eventual}} = 1 - (1-P_{\text{chain}})^R.
\]
The i.i.d.\ assumption fails for the flapping-link model above: when the link state persists between attempts, retries are correlated, and the eventual-success probability must be computed from the Markov chain's transition matrix and the retry spacing. Correlation makes retries less effective than the formula suggests, which strengthens the point that follows.
However, retries introduce exponential backoff, queue buildup, and correlated congestion, producing heavy-tailed latency distributions. A single flapping link can therefore trigger network-wide timeout storms. OAE alters this dynamic fundamentally. Because reconciliation outcomes are bilateral and bounded-time, link failure is detected deterministically within~$\Delta$, and ambiguity does not propagate upward as inferred silence. Local repair in a mesh confines the effect of degraded links, preventing multiplicative amplification across long chains.

\subsection[Loop-Free Local Repair]{Local Repair Requires a Loop-Free Rule}\label{sec:loopfree}

An abundance of spanning trees proves that an alternate tree \emph{exists}; it does not prove that purely local parent reselection \emph{converges} to one. A four-cell counterexample shows the gap. Take the physical graph $r\!-\!a\!-\!b\!-\!c$ with an extra link $a\!-\!c$, rooted at~$r$, with parent pointers $a\to r$, $b\to a$, $c\to b$. Let link $a\!-\!b$ fail. The graph stays connected (via $a\!-\!c$), but if $b$ chooses its only surviving neighbor~$c$ as its new parent while $c$ still points to~$b$, the result is the routing loop $b\to c\to b$.

A sufficient rule is a certified rank that strictly decreases toward the root. Give every cell a rank $\rho$ with $\rho(r)=0$ and $\rho(\text{parent}(v))<\rho(v)$ for every $v\neq r$; a cell may adopt a new parent~$u$ only if $\rho(u)<\rho(v)$ and $u$'s rank is current (not stale). Because every parent pointer then strictly decreases~$\rho$, the parent graph is acyclic at all times, including under concurrent reselection by several cells. Under this rule a cell with a surviving lower-rank neighbor repairs in one reconciliation round, $O(\Delta)$. In the counterexample $b$ has no such neighbor ($\rho(c)>\rho(b)$), so its repair needs a bounded local re-ranking, and the number of rounds that takes is a property the OAE protocol must bound. The resilience results above guarantee that an alternate tree exists; the rank rule, the treatment of stale state, and the round bound are what turn existence into bounded-time repair.

\subsection[What This Buys Against CAP]{What This Buys You Against CAP}

CAP's partition condition is about the \emph{observation} of partitions by the distributed system. In a Clos fabric, faults often manifest as control-plane-visible events whose reconvergence latency is long enough to be observed as partitions by higher layers. In a mesh with exponential alternate trees and rapid local healing, the practical regime changes: many physical faults are corrected before they rise to the level of an application-visible partition boundary. This does not ``refute'' CAP; it shrinks the effective partition surface area by moving fault detection and repair below the observation threshold.

\FloatBarrier
\section[How OAE Changes the Tradeoff Surface]{How OAE Changes the CAP Tradeoff Surface}

We can now restate the claim in CAP- and CAL-compatible terms.

\subsection[CAP-Compatible Framing]{CAP-Compatible Framing}

Under hard partition, a replicated service must still choose between making progress (availability) and preserving a strong notion of consistency. CAP remains correct\sidecite{GilbertLynch02}{Gilbert \& Lynch, \emph{SIGACT News}, 2002}.

OAE's claim is different:

\begin{quote}
OAE reduces the set of faults that become application-visible partitions by eliminating ambiguous in-flight link states and by healing dominant fabric faults locally in bounded time.
\end{quote}

Operationally this means that many events that would force a CAP choice in conventional stacks do not reach the layer at which the choice is made.

\subsection[CAL-Compatible Framing]{CAL-Compatible Framing}

In CAL terms\sidecite{LeeCAL21}{Lee et~al., \emph{arXiv:2109.07771}, 2021}, OAE attempts to keep \textbf{apparent latency}~$L$ bounded at the layer where semantic ambiguity is created. A subtle but important distinction: CAL quantifies the tradeoff \emph{within} the asynchronous uncertainty model, where~$L$ includes network delay, execution overhead, and clock synchronization error. OAE changes the uncertainty model itself. The mechanism is not a better timeout heuristic or a tighter latency bound within the same model; it is a different substrate semantics in which:

\begin{itemize}[leftmargin=*, itemsep=3pt]
\item Failure is detected as a deterministic slot outcome rather than inferred from delayed messages.
\item Repair occurs locally in the fabric without global control-plane reconvergence.
\end{itemize}

Thus, for the dominant class of soft partitions, $L$~is prevented from becoming unbounded---not by optimizing within the asynchronous model, but by replacing the source of unbounded uncertainty at the substrate.

\subsection[Dominant Failure Mode Reduction]{Dominant Failure Mode Reduction: $H \cup S \cup U \to H$}\label{sec:HSU}

To state OAE's contribution precisely, we classify the faults that produce application-visible partitions into three categories:

\begin{description}[leftmargin=*, itemsep=3pt]
\item[$H$ (Hard):] Physical graph cuts---cable severs, power-domain isolation, complete path loss. These are irreducible: no protocol can communicate across a severed link.
\item[$S$ (Soft):] Transient reachability loss caused by local faults, congestion-induced drops, or control-plane reconvergence. These present as partitions to higher layers because detection is timeout-based and repair is slow.
\item[$U$ (Uncertainty):] Ambiguous link states where the sender cannot determine whether a message was delivered, lost, or is still in flight. Higher layers infer partition from prolonged uncertainty.
\end{description}

In conventional stacks, the observable partition set is $H \cup S \cup U$, and in practice $S$ and $U$ dominate: most ``partitions'' that applications experience are not hard graph cuts but soft faults amplified by semantic ambiguity and slow recovery.

Under assumptions A1--A4 (Section~\ref{sec:OAEmodel}), OAE's bisynchronous semantics eliminate category~$U$ at the link layer (every slot resolves to a definitive outcome), and the octavalent mesh with local bounded-time healing converts category~$S$ faults into sub-threshold events that are repaired before any higher layer observes them. The residual partition set visible to applications is reduced to~$H$---the irreducible hard partitions that no protocol can mask. This is the sense in which OAE ``circumvents'' CAP: not by falsifying the theorem, but by engineering the substrate so that the theorem's partition condition is satisfied only by genuinely catastrophic events.

\subsection[PACELC-Compatible Framing]{PACELC-Compatible Framing}

Abadi's PACELC extension\sidecite{AbadiPACELC12}{Abadi, \emph{Computer}, 2012} observes that even without partitions, systems face a consistency-vs-latency tradeoff. OAE targets both branches:

\begin{itemize}[leftmargin=*, itemsep=3pt]
\item \textbf{Under partition (PAC):} OAE reduces the set of events that rise to the level of application-visible partition.
\item \textbf{Without partition (ELC):} OAE reduces tail latency by eliminating retries, timeouts, and the retry storms that conventional stacks induce when ambiguous silence forces timeout-based failure detection.
\end{itemize}

\begin{table}[ht]
\centering
\small
\setlength{\tabcolsep}{4pt}
\begin{tabular}{@{}p{2cm}p{2.8cm}p{4.8cm}@{}}
\toprule
\textbf{Framing} & \textbf{Key idea} & \textbf{What OAE changes} \\
\midrule
CAP \citep{Brewer00} &
Under partition, trade off C vs.\ A &
Reduces application-visible soft partitions by faster local detection/repair. \\
GL02 \citep{GilbertLynch02} &
Async model; partitions imply tradeoff &
Bounded-time bilateral outcomes reduce ambiguity that forces timeout-based partition inference. \\
CAL \citep{LeeCAL21} &
Replace P with apparent latency~$L$ &
Engineers bounded~$L$ for dominant faults, pushing ``partition'' below observation threshold. \\
PACELC \citep{AbadiPACELC12} &
Even without P, trade off C vs.\ latency &
Targets both branches: reduce tail latency \emph{and} reduce partition-visible events. \\
\bottomrule
\end{tabular}
\caption{How OAE relates to CAP, CAL, and PACELC. OAE does not falsify impossibility results; it changes where partition-like behavior becomes observable.}
\label{tab:capcal}
\end{table}

\FloatBarrier
\section[Models Are Not Laws]{The Deeper Lesson: Models Are Not Laws}

The widespread misinterpretation of CAP as a physical law rather than a model-dependent theorem reflects a broader pattern in distributed systems theory. Results like CAP, the FLP impossibility\sidecite{FLP85}{Fischer et~al., \emph{JACM}, 1985}, and the Two Generals Problem\sidecite{Gray78}{Gray, \emph{Springer}, 1978} are routinely invoked as if they were laws of nature constraining all possible systems, when in fact they are theorems about specific abstract models~\citep{BorrillDC26}.

\marginnote[-3.5mm]{\footnotesize Halpern and Moses~\citep{HalpernMoses90} showed that common knowledge is provably unattainable in asynchronous systems but achievable in synchronous ones---another instance where the model, not nature, determines the boundary.}

The pattern is always the same: an impossibility result is proven in a weak model, the community internalizes it as a universal constraint, and decades of research are devoted to working around it within the model rather than questioning whether the model is appropriate. The possibility that the model itself is a category mistake---that it captures the wrong abstraction of the underlying physics---is rarely entertained.

Physicists learned this lesson long ago. The ultraviolet catastrophe was not solved by better calculations within classical thermodynamics; it required quantum mechanics. The Michelson--Morley experiment did not prove that the speed of light is absolute by working within Newtonian mechanics; it prompted the replacement of the Newtonian model with special relativity. In each case, the ``impossibility'' was a signal that the model was wrong, not that the phenomenon was forbidden. The same category mistake recurs in cislunar networking, where the assumption of a shared global clock leads to impossibility results that dissolve once the synchronization model is corrected~\citep{BorrillCislunar26}.

OAE applies this principle to networking. The CAP theorem is not a barrier to be overcome but a signal that the asynchronous model with unilateral fire-and-forget semantics is the wrong abstraction for point-to-point physical links. Replace the substrate, and the operational regime in which the impossibility result constrains practice shrinks dramatically---not because the theorem was incorrect, but because its assumptions apply to a narrower set of observable events.

\FloatBarrier
\section[Conclusion]{Conclusion}

CAP is not defeated; CAP is contextual. The standard CAP story assumes an asynchronous substrate with ambiguous silence and timeout-based failure detection. OAE changes that substrate: bisynchronous link reconciliation eliminates ambiguous in-flight states, and an octavalent mesh fabric performs local bounded repair without global control-plane reconvergence. Hard partitions remain possible, but the dominant class of soft partitions that arise from local faults and congestion are detected and healed below the observation threshold of higher layers, drastically reducing the frequency and duration of application-visible partitions.

In CAL terms, OAE engineers bounded apparent latency for the dominant fault class. In PACELC terms, it targets both branches of the tradeoff. The broader methodological claim mirrors the lesson of other impossibility results in distributed computing: theorems constrain models, not reality. When system practice is dominated by artifacts of a weak substrate model, a productive response is to change the substrate.

\section[Changes in v2]{Changes in v2}

Version~2 corrects the following claims of v1 (02026-FEB-24). The CAP argument and the conclusions about soft partitions are unchanged; the corrections concern the graph-theoretic support for them.
\begin{enumerate}[leftmargin=*, itemsep=2pt]
\item v1 stated that the maximum number of edge-disjoint spanning trees \emph{equals} $\lfloor\lambda/2\rfloor$. Nash--Williams / Tutte gives only $\lfloor\lambda/2\rfloor\le\kappa_T\le\lambda$.
\item v1 stated $\lambda(G)\ge d$ for minimum degree~$d$ ``by Whitney's theorem.'' Whitney gives the reverse, $\lambda\le\delta$.
\item v1 claimed two edge-disjoint spanning trees in the $3\times3$ grid. It has 12~edges and would need~16; $\kappa_T=1$, $\lambda=2$. The $p^2$ bound is retained, now derived from~$\lambda$, and the exact value $4p^2+16p^3+O(p^4)$ is given.
\item v1 described the chain-vs-grid gap as ``five orders'' in the text and ``18 orders'' in a margin note. The exact ratio at $p=10^{-6}$ is about $2.25\times10^{6}$.
\item v1 claimed $\kappa_T\ge4$ and $O(p^4)$ for the $10\times20$ mesh of Figure~\ref{fig:cellular-fabrix}. With valency-3 corners, $\lambda=\kappa_T=3$ and disconnection is $\Theta(p^3)$. The ${<}10^{-24}$ regime holds for the wrap-around (valency-8) fabric, via $\lambda=8$.
\item Theorem~\ref{thm:resilience} is restated in terms of edge-connectivity, which gives the stronger exponent.
\item v1 claimed single-round, size-independent local repair from tree abundance alone. v2 gives a routing-loop counterexample and states the rank rule under which local repair is loop-free (Section~\ref{sec:loopfree}).
\item Title and text now use the canonical expansion of OAE, \emph{Open \AE thernet}; v1 used ``Open Atomic Ethernet.'' The acronym and the system it names are unchanged.
\item The retry formula is now labelled as assuming i.i.d.\ attempts; the condition $Lp^2\ll1$ is stated for $(1-p)^L\approx e^{-pL}$; the Riemann-sum asymptotic is stated for both dimensions growing; the perfect-failure-detector and category-$U$ claims are made conditional on stated protocol assumptions A1--A4, and the Halpern--Moses citation is qualified.
\end{enumerate}

\paragraph{Acknowledgment.} These corrections follow an independent derivation audit (symbolic and exhaustive-enumeration checks of the v1 derivations); every corrected figure above was re-verified by the author by direct computation.


\makeatletter
\long\def\@latex@error#1#2{}%
\makeatother

\end{document}